\documentclass[11pt,twoside]{article}

\usepackage{pgfplots}
\pgfplotsset{compat=newest}
\usepackage{tikz}

\usepackage{fancyhead}
\usepackage{psfrag}
\usepackage{epsfig}
\usepackage{graphicx}
\usepackage{amsmath}
\usepackage{amssymb}
\usepackage{amsthm}
\usepackage{subfigure}
\usepackage{cite}
\usepackage{color}
\usepackage{lipsum}
\usepackage[latin1]{inputenc}
\newcommand{\Ya}{R}

\newtheorem{lemma}{Lemma}
\newtheorem{theorem}{Theorem}

\def\ve#1{{\mathchoice{\mbox{\boldmath$\displaystyle #1$}}%
              {\mbox{\boldmath$\textstyle #1$}}%
              {\mbox{\boldmath$\scriptstyle #1$}}%
              {\mbox{\boldmath$\scriptscriptstyle #1$}}}}

\renewcommand{\i}{\ve{i}}
\renewcommand{\it}{\i_t}
\newcommand{\itm}{\i_{t-1}}

\newcommand{\itstar}{\it^{\ast}}
\newcommand{\itdiam}{\it^{\diamond}}
\newcommand{\itmstar}{\itm^{\ast}}

\newcommand{\0}{\ve{0}}
\newcommand{\G}{\ve{G}}
\newcommand{\Gt}{\G_0}
\newcommand{\Gtm}{\G_1}
\newcommand{\Gtstar}{\Gt^{\ast}}
\newcommand{\Gtdiam}{\Gt^{\diamond}}
\newcommand{\Gtmstar}{\Gtm^{\ast}}

\newcommand{\C}{\mathcal{C}}
\newcommand{\Calpha}{\C_{\alpha}}
\newcommand{\Czero}{\C_{0}}
\newcommand{\Cone}{\C_{1}}
\newcommand{\Czeroone}{\C_{01}}

\renewcommand{\c}{\ve{c}}
\newcommand{\ct}{\ve{c}_t}
\newcommand{\ctm}{\ve{c}_{t-1}}
\newcommand{\ctp}{\ve{c}_{t+1}}
\renewcommand{\r}{\ve{r}}
\newcommand{\e}{\ve{e}}

\newcommand{\dist}{\mathrm{d}}
\newcommand{\wt}{\mathrm{wt}}
\newcommand{\taualpha}{\tau_{\alpha}}
\newcommand{\tauzero}{\tau_{0}}
\newcommand{\tauone}{\tau_{1}}
\newcommand{\tauzeroone}{\tau_{01}}

\newcommand{\pa}{p_a}
\newcommand{\pb}{p_b}
\newcommand{\pc}{p_c}
\newcommand{\pd}{p_d}

\renewcommand{\L}{L}
\newcommand{\kone}{k_1}
\newcommand{\F}{\mathbb{F}}

\newcommand{\RowSpace}[1]{\left\langle #1 \right\rangle}

\newcommand{\X}{X}
\newcommand{\Prob}{\mathrm{P}}
\newcommand{\Qt}{Q_t}
\newcommand{\Qtm}{Q_{t-1}}
\newcommand{\Qone}{Q_1}
\newcommand{\Yat}{\Ya_t}
\newcommand{\Yatp}{\Ya_{t+1}}
\newcommand{\YaL}{\Ya_{\L}}

\newcommand{\Pt}{P_t}
\newcommand{\Pindt}{P_{t,\mathrm{ind}}}

\newcommand{\BigOtext}[1]{O(#1)}

\begin{document}

\vspace*{5mm}

\noindent
\textbf{\LARGE On the Success Probability of Decoding \\(Partial) Unit Memory Codes
}
\thispagestyle{fancyplain} \setlength\partopsep {0pt} \flushbottom
\date{}

\vspace*{5mm}
\noindent
\textsc{Sven Puchinger 
	} \hfill \texttt{sven.puchinger@uni-ulm.de} \\
\textsc{Sven Müelich} \hfill \texttt{sven.mueelich@uni-ulm.de} \\
\textsc{Martin Bossert} \hfill \texttt{martin.bossert@uni-ulm.de} \\
{\small Institute of Communications Engineering, University of Ulm, Germany}

\medskip

\begin{center}
\parbox{11,8cm}{\footnotesize
\textbf{Abstract.}
In this paper, we derive analytic expressions for the success probability of decoding (Partial) Unit Memory codes in memoryless channels. 
An applications of this result is that these codes outperform individual block codes in certain channels.
}
\end{center}

\baselineskip=0.9\normalbaselineskip

\section{Introduction}

\addtolength\belowdisplayskip{-0.5\baselineskip}%
\addtolength\abovedisplayskip{-0.5\baselineskip}%

(Partial) Unit Memory ((P)UM) codes, introduced in \cite{Lee1976} and \cite{Lauer1979}, are convolutional codes, defined using block codes.
Several (P)UM code constructions and a decoder based on the underlying block codes were proposed in \cite{Dettmar1993,Dettmar1994,Dettmar1995}.
Since these publications, there have been results on improving the decoding algorithm \cite{PWB2014}, extension to rank-metric codes \cite{WSS2014}, and applications to random linear network coding \cite{PCFBH2015} and the streaming scenario \cite{KB2016}.
In these applications, the codes were evaluated numerically in probabilistic channels and it was observed that (P)UM codes often outperform individual block codes in these scenarios.

In this paper, we derive analytic expressions for the probability of successfully recovering an
information block that is encoded with a (P)UM code, in memoryless channels.
Using these new expressions, we are able to partly explain the numerical observations in \cite{PCFBH2015} and \cite{KB2016} analytically.

\section{(Partial) Unit Memory Codes}

\lhead{} \rhead{}
\chead[\fancyplain{}{\small\sl\leftmark}]{\fancyplain{}{\small\sl
\leftmark}} \markboth{\hspace{-2,15cm}Eighth International Workshop
on Optimal Codes and Related Topics\\ July 10-14, 2017, Sofia, Bulgaria \hfill pp. xxx-xxx}{}

We use the description of (P)UM codes as in \cite{Dettmar1994}.
Let $k \leq n$ and $\kone \leq \min\{k,n-k\}$ be non-negative integers.
We choose matrices $\Gt$ and $\Gtm$ of the form
\begin{equation*}\small
\Gt = \begin{bmatrix} \Gtstar \\ \Gtdiam \end{bmatrix}, \quad \Gtm = \begin{bmatrix} \Gtmstar \\ \0 \end{bmatrix},
\end{equation*}
where the row spaces of the three matrices $\Gtstar, \Gtmstar \in \F^{\kone \times n}$ and $\Gtdiam \in \F^{k-\kone \times n}$ pairwise intersect only in the zero codeword.
Let $\RowSpace{\G}$ denote the row space of a matrix $\G$. We define the following codes:
\begin{equation*}\small
\Calpha := \RowSpace{
\begin{bmatrix}
\Gtstar \\
\Gtdiam \\
\Gtmstar
\end{bmatrix}
}, \quad
\Czero := \RowSpace{
\begin{bmatrix}
\Gtstar \\
\Gtdiam
\end{bmatrix}
}, \quad
\Cone := \RowSpace{
\begin{bmatrix}
\Gtdiam \\
\Gtmstar
\end{bmatrix}
}, \quad
\Czeroone := \RowSpace{\Gtdiam}.
\end{equation*}
By the properties of $\Gtstar, \Gtdiam, \Gtmstar$, the codes have parameters
\begin{equation*}
\Calpha(n,k+\kone), \quad \Czero(n,k), \quad \Cone(n,k), \quad \Czeroone(n,k-\kone).
\end{equation*}
Since $\Czeroone \subseteq \Cone \subseteq \Calpha$ and $\Czeroone \subseteq \Czero \subseteq \Calpha$, the error correction capability of $\Czeroone$ is typically the largest, followed by both $\Czero$ and $\Cone$. $\Calpha$ is the weakest code in terms of error correction.

\subsection{Encoding}

Given the generator matrices $\Gt$ and $\Gtm$, we encode a sequence of information vectors $\it \in \F^k$ ($t=0,\dots,\L$, where we choose $\i_0 = \i_\L = \0$) into a code sequence
\begin{equation*}
\c_t = \it \cdot \Gt + \itm  \cdot \Gtm \quad \text{for } t=1,\dots,\L.
\end{equation*}
Note that we can re-write this relation into
\begin{equation}
\c_t = \itstar \cdot \Gtstar + \itdiam \cdot \Gtdiam + \itmstar \cdot \Gtmstar. \label{eq:encoding}
\end{equation}
The information vectors and codewords corresponding to an index $t$ are called $t$-th \emph{block}. If $k<\kone$, the resulting code is called $(n,k|\kone)$ \emph{partial unit memory code} since the current codeword $\ct$ contains parts of the previous information word $\itm$.
If $k=\kone$ (note that we require $k\leq \tfrac{n}{2}$ in this case), the code is called $(n,k)$ \emph{unit memory} code since $\ct$ depends on the entire information vector $\itm$ (in particular, $\itm$ can be completely recovered by knowing $\ct$ or $\ctm$).

\subsection{Decoding}

\pagestyle{myheadings} \markboth{OC2017}{Puchinger, M\"uelich, Bossert}

The sequence of received words is of the form 
\begin{align*}
\r_t = \c_t + \e_t \quad \text{for } t=1,\dots,\L.
\end{align*}
We choose some\footnote{This can e.g.\ be the Hamming metric as in \cite{Dettmar1993} and \cite{KB2016}, or the rank metric as in \cite{WSS2014}.} metric $\dist(\cdot,\cdot) : \F^n \times \F^n \to \mathbb{R}_{\geq 0}$, and the corresponding weight $\wt(\cdot) = \dist(\cdot,\0)$, for which we know decoders of the codes
\begin{align*}
\Calpha,\Czero,\Cone,\Czeroone
\end{align*}
that can find all codewords with distance to the received word at most
\begin{align*}
\taualpha,\tauzero,\tauone,\tauzeroone,
\end{align*}
respectively.
We assume that $\taualpha < \tauzero = \tauone < \tauzeroone$ in this paper.\footnote{This is not a major restriction since most known PUM constructions, e.g.\ based on Reed--Solomon, BCH \cite{Dettmar1993}, or Gabidulin codes \cite{WSS2014}, provide codes $\Czero$, $\Cone$ of the same minimum distance.}
For notational convenience, we say that \emph{$t$ errors occurred} if the error word has weight $t$.

We use the description of decoding as in \cite{Dettmar1995}.
There, the Hamming metric in combination with bounded-minimum-distance decoders was used.
However, the decoder also works with list decoders in the Hamming metric \cite{PWB2014}, with rank-metric PUM codes \cite{WSS2014}, or with erasures \cite{KB2016}.
First, candidates for the codewords $\c_t$ are found in $4$ steps (see below). Afterwards, the most likely sequence $\c_1,\dots,\c_\L$ is found among these candidates using the Viterbi algorithm.
In this paper, we say that decoding is successful at the $t$-th position if the sent codeword $\c_t$ is among the candidates.
Finding the candidates works in $4$ steps:
\begin{enumerate}
\item Each received word $\r_t = \c_t + \e_t$ is decoded independently using the decoder of $\Calpha$ (note $\c_t \in \Calpha$). We can decode up to $\taualpha$ errors in this step.
\item Using the information fragment $\itmstar$ given by a candidate codeword $\ctm$, we can successfully decode the right neighbor $\ct$ in the code $\Czero$ if $\wt(\e_t)\leq \tauzero$, using the following relation (note that the left-hand side is known)
\begin{equation*}
\r_t - \itmstar \cdot \Gtmstar = \underset{\in \, \Czero}{\underbrace{\itstar \cdot \Gtstar + \itdiam \cdot \Gtdiam}} + \e_t.
\end{equation*}
We can repeat this so-called \emph{forward} step iteratively for all candidates.
\item Similar to Step~$2$, we can go in \emph{backward} direction by decoding
\begin{equation*}
\r_t - \itstar \cdot \Gtstar = \underset{\in \, \Cone}{\underbrace{\itdiam \cdot \Gtdiam + \itmstar \cdot \Gtmstar}} + \e_t,
\end{equation*}
in the code $\Cone$, which is successful if the number of errors is at most $\tauone$.
\item Using $\Czeroone$, we can find the $\c_t$ in positions $t$, where both neighbor blocks $t-1$ and $t+1$ have been successfully decoded and $\wt(\e_t) \leq \tauzeroone$, using
\begin{equation*}
\r_t - \itstar \cdot \Gtstar - \itmstar \cdot \Gtmstar = \underset{\in \, \Czeroone}{\underbrace{\itdiam \cdot \Gtdiam}} + \e_t,
\end{equation*}
\end{enumerate}

\section{New Expressions for the Success Probability}
\label{sec:new_bounds}

Let the PUM code and constituent decoders with decoding radii $\taualpha,\tauzero,\tauone,\tauzeroone$ be given.
We assume that the error words $\e_t$ are drawn i.i.d.\ at random according to an arbitrary distribution (memoryless channel).
Let $\X_1,\dots,\X_\L$ be the random variables describing the error weight, i.e., $\X_t := \wt(\e_t)$.
Thus, the $\X_t$ are also independently and identically distributed as some random variable $\X$.

In the following, we derive an expression for the probability
\begin{align*}
\Pt = \Prob(\i_t \text{ is found})
\end{align*}
that the $t$-th information word $\i_t$ of the PUM code is successfully recovered (i.e., among the candidates), only depending on the distribution of $\X$ and the position $t$.
The expression depends on the probabilities
\begin{align*}
\pa &:= \Prob(0 \leq X \leq \taualpha), &&\pb := \Prob(\taualpha < X \leq \tauzero), \\
\pc &:= \Prob(\tauzero < X \leq \tauzeroone), &&\pd := \Prob(\tauzeroone < X).
\end{align*}
Note that $\pa+\pb+\pc+\pd=1$.
Let $\Qt$ denote the probability that the $t$-th block is correctly decoded by Step~$1$ or~$2$ (individually or in forward direction).
Similarly, by $\Yat$ we define the probability that it is found by Step~$1$ or~$3$.
\begin{lemma}\label{lem:QtYat}
For all $t=1,\dots,\L$, we have
\begin{equation*}
\Qt = \tfrac{\pa}{1-\pb} + \pb^t\cdot \left(\tfrac{1-\pa-\pb}{1-\pb} \right) \text{ and } \, \Yat = \tfrac{\pa}{1-\pb} + \pb^{\L-t+1}\cdot \left(\tfrac{1-\pa-\pb}{1-\pb} \right).
\end{equation*}
\end{lemma}

\begin{proof}
We prove the claim by induction. Since the information word $\i_0 = \0$ is known, we can directly decode the first codeword $\c_1$ in $\Czero$ and obtain
\begin{equation*}
\Qone = \Prob(\X_1 \leq \tauzero) = \pa+\pb = \tfrac{\pa+\pb-\pa\pb-\pb^2}{1-\pb} = \tfrac{\pa}{1-\pb} + \pb^1\cdot \left(\tfrac{1-\pa-\pb}{1-\pb} \right).
\end{equation*}
The probability that the $t$-th block is found in forward direction is given by the sum of the probability that it is found individually and the probability that Step~$1$ fails, but it is successfully recovered in forward direction, i.e.,
\begin{align*}
\Qt &= \Prob(\X_t \leq \taualpha) + \Prob(\taualpha < \X_{t} \leq \tauzero) \cdot \Qtm \\
&= \pa + \pb \cdot \left(\tfrac{\pa}{1-\pb} + \pb^{t-1} \cdot \left(\tfrac{1-\pa-\pb}{1-\pb} \right)\right) = \tfrac{\pa}{1-\pb} + \pb^t\cdot \left(\tfrac{1-\pa-\pb}{1-\pb} \right).
\end{align*}
The proof of for $\Yat$ is equivalent using the base case $\YaL = \pa+\pb$.
\end{proof}
Note that for all $t,L$, the probabilities $\Qt,\Yat$ are lower-bounded by $\tfrac{\pa}{1-\pb}$.

\subsection{Partial Unit Memory Codes ($\kone<k$)}
\label{subsec:new_bounds_pum}

In the case of PUM codes, the correct information vector $\it$ is found if and only if $\ct$ is found.
Hence, we can state the following result.
\begin{theorem}\label{thm:PUM_Pt}
For any $t=1,\dots,\L$, we have
\begin{equation*}
\Pt = \pa + \tfrac{\pa}{(1-\pb)^2} \big[\pb(2-\pa-2\pb) + \pa\pc \big] + \varepsilon(t,\L),
\end{equation*}
where $\varepsilon(t,\L) \geq 0$ and $\varepsilon(t,\L) \in \BigOtext{\max\{ \pb^t,\pb^{L-t} \}}$ (i.e., the term $\varepsilon(t,\L)$ is negligible if $t$ is sufficiently far away from $0$ and $\L$).
\end{theorem}

\begin{proof}
We can write
\begin{align*}
\Pt &= \Prob(\i_t \text{ is found}) = \Prob(\c_t \text{ is found}) \\
&= \underset{\text{found in Step~$1$}}{\underbrace{\Prob(\X_t \leq \taualpha)}} + \underset{\text{found only in forward or backward direction}}{\underbrace{\Prob(\taualpha < \X_t \leq \tauzero) \cdot (\Qtm + \Yatp - \Qtm \Yatp)}}  \\
& \quad \quad + \underset{\text{found in Step } 4}{\underbrace{\Prob(\tauzero < \X_t \leq \tauzeroone) \cdot \Qtm \cdot \Yatp}} \\
&= \pa + \pb (\Qtm + \Yatp - \Qtm \Yatp) + \pc \Qtm \Yatp.
\end{align*}
Let $A := \tfrac{\pa}{1-\pb}$, $B := \tfrac{1-\pa-\pb}{1-\pb} \pb^{t-1}$, and $C:= \tfrac{1-\pa-\pb}{1-\pb} \pb^{L-t}$. Then, $A,B,C \geq 0$ and $B \in O(\pb^t)$ and $C \in O(\pb^{\L-t})$.
Also, $\Qtm = A+B$ and $\Yatp = A+C$, so
\begin{align*}
\Pt &= \pa + \pb (A+B+A+C-(A+B)(A+C)) + \pc (A+B)(A+C) \\
&= \pa + \tfrac{\pa}{(1-\pb)^2} \big[\pb(2-\pa-2\pb) + \pa\pc \big] + \varepsilon(t,\L), \text{ where} \\
\varepsilon(t,\L) &= \pb \underset{\geq (A+B)(B+C)-A(B+C)-BC = B^2 \geq 0}{\underbrace{(B+C-AB-AC-BC)}} + \pc (AB+AC+BC) \geq 0.
\end{align*}
Since all terms depend on $B$ or $C$, we have $\varepsilon(t,\L) \in O(\max\{ \pb^t,\pb^{L-t}\})$.
\end{proof}

\subsection{Unit Memory Codes ($k=\kone$)}
\label{sub_sec:new_bounds_um}

Unit memory codes have the advantage that we can obtain $\it$ from either $\ct$ or~$\ctp$.
In UM codes, $\Czeroone$ has dimension $k-\kone=0$, and hence, Step~$4$ is not useful.
On the other hand, we can define $\tauzeroone := \infty$, so $\pd = 0$ and $\pc = 1-\pa-\pb$.

\begin{theorem}\label{thm:UM_Pt}
For any $t$, there is a $\delta(t,\L) \geq 0$ with $\delta(t,\L) \in \BigOtext{\max\{ \pb^t,\pb^{L-t} \}}$:
\begin{equation*}
\Pt = 1-\left(\tfrac{\pc}{1-\pb}\right)^2 + \delta(t,\L).
\end{equation*}
\end{theorem}

\begin{proof}
We define $A,B,C$ as in the proof of Theorem~\ref{thm:PUM_Pt}. Then, we can write
\begin{align*}
\Pt &= \Qt + \Yatp - \Qt \Yatp = A+B+A+C+(A+B)(A+C) \\
&= A(2-A) + \underset{=: \, \delta(t,\L)}{\underbrace{B+C-A(B+C)-BC}} = 1-\left(\tfrac{\pc}{1-\pb}\right)^2 + \delta(t,\L),
\end{align*}
where $\delta(t,\L)$ has the desired properties.
\end{proof}

\section{Applications}

\subsection{Fast Code Design}

Based on the results in Section~\ref{sec:new_bounds}, it is possible to determine the failure probability of decoding a PUM code block only from the probability density function (pdf) of the error weight in a block (which is given by the channel and the block length $n$).
Hence, as soon as this pdf is determined (either theoretically or numerically), one can compute the decoding failure probability for any code parameter set, i.e., variations of $k$, $\kone$, $\taualpha$, $\tauzero$, $\tauone$, and $\tauzeroone$, without the need for computationally expensive Monte-Carlo simulations.
This allows to optimize code parameters (e.g. $\kone$ for given $k$) quickly.

Figure \ref{fig:monte_carlo_vs_exact} shows the results of a Monte-Carlo simulation compared to the exact failure probability expressions from Section~\ref{sec:new_bounds}. As expected, the resulting curves coincide up to the estimation error of the Monte-Carlo simulation.

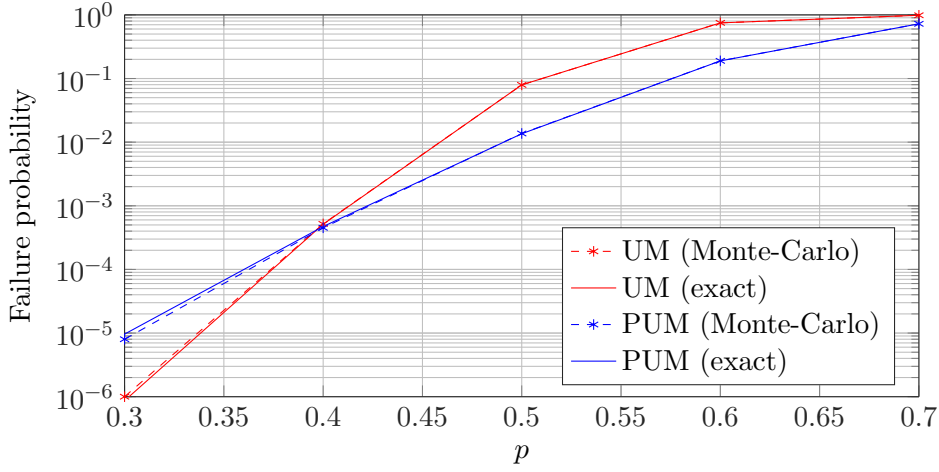
\begin{figure}[h]
%
%
\begin{tikzpicture}

\begin{axis}[%
width=0.8\textwidth,
height=2in,
scale only axis,
separate axis lines,
every outer x axis line/.append style={white!15!black},
every x tick label/.append style={font=\color{white!15!black}},
xmin=0.3,
xmax=0.7,
xmajorgrids,
xlabel={$p$},
every outer y axis line/.append style={white!15!black},
every y tick label/.append style={font=\color{white!15!black}},
ymode=log,
ymin=1e-6,
ymax=1,
yminorticks=true,
ymajorgrids,
yminorgrids,
ylabel={Failure probability},
legend style={at={(0.97,0.03)},anchor=south east,draw=white!15!black,fill=white,legend cell align=left}
]
\addplot [color=red,dashed,mark=asterisk,mark options={solid}]
  table[row sep=crcr]{%
0.2	0\\
0.3	1e-06\\
0.4	0.000519\\
0.5	0.079397\\
0.6	0.749582\\
0.7	0.986132\\
};
\addlegendentry{UM (Monte-Carlo)};

\addplot [color=red,solid]
  table[row sep=crcr]{%
0.2	1.76140679265821e-10\\
0.3	8.66190499355916e-07\\
0.4	0.000513363384624099\\
0.5	0.0794773140435064\\
0.6	0.748684477113017\\
0.7	0.985978167938645\\
};
\addlegendentry{UM (exact)};

\addplot [color=blue,dashed,mark=asterisk,mark options={solid}]
  table[row sep=crcr]{%
0.2	0\\
0.3	8e-06\\
0.4	0.000451\\
0.5	0.013687\\
0.6	0.189529\\
0.7	0.723939\\
};
\addlegendentry{PUM (Monte-Carlo)};

\addplot [color=blue,solid]
  table[row sep=crcr]{%
0.2	5.73586115271141e-08\\
0.3	9.63108851848027e-06\\
0.4	0.000472349494225033\\
0.5	0.0135612992270332\\
0.6	0.189039748875718\\
0.7	0.724074460091017\\
};
\addlegendentry{PUM (exact)};
\end{axis}
\end{tikzpicture}%
\caption{Estimated failure probability by Monte-Carlo simulation with $10^6$ samples (solid lines) compared to the exact failure probability expressions $\Pt$ from Section~\ref{sec:new_bounds} of a $(15,5|2)$ PUM code ($\taualpha = 8$, $\tauzero=\tauone=10$, $\tauzeroone=12$) and a $(15,5)$ UM code  ($\taualpha = 5$, $\tauzero=\tauone=10$) with $L=100$ and $t=50$. The error weight in each block is independently binomially distributed with parameters $n=15$ and $p$ (varying parameter).}
\label{fig:monte_carlo_vs_exact}
\end{figure}

\subsection{(P)UM Codes vs.\ Independent Block Codes}
\label{subsec:PUM_vs_MDS}

Let $\C(n,k)$ be a linear block code with the same rate as the PUM code. For a fair comparison, we assume that decoding in $\C$ is possible up to $\tauzero$ (i.e., as the decoding radius of the code $\Czero$ in the PUM coding scheme).
Consider a channel in which a position independently adds $1$ to the error weight of the block with probability $p$ (i.e., the error weight is binomially distributed with parameters $n$ and $p$).

In this section, we show that for $p \to 0$, the failure probability of decoding (P)UM codes gets below the one of encoding/decoding each information block independently in $\C$ (i.e., for a generator matrix $\G$ of $\C$, we have $\ct = \it \cdot \G$ for all $t$).
This fact was observed before in numerical simulations, e.g., in \cite{PWB2014} or \cite{KB2016}, but no theoretical explanation was known.
We require the following observations.

\begin{lemma}[\!\!{\cite[page~115]{Ash1990}}]
\label{lem:tail_bounds}
For $\tau > pn$, we have
\begin{align*}
\Prob(X \geq \tau) \begin{cases}
\leq \exp\left( - \tau \log\big(\tfrac{\tau}{n p}\big) - (n-\tau) \log\big(\tfrac{n-\tau}{n (1-p)}\big) \right) \\
\geq \tfrac{1}{\sqrt{2n}}\exp\left( - \tau \log\big(\tfrac{\tau}{n p}\big) - (n-\tau) \log\big(\tfrac{n-\tau}{n (1-p)}\big) \right)
\end{cases}
\end{align*}
\end{lemma}

Lemma~\ref{lem:tail_bounds} implies that $\Prob(X \geq \tau) \in \Theta(p^\tau (1-p)^{n-\tau})$ when considering $n$ and $\tau$ as constants. For $p \to 0$, this means that $\Prob(X \geq \tau) \in \Theta(p^\tau)$, which implies the following lower and upper bounds on the probabilities $\pa$,...,$\pd$.
\begin{lemma}
For $p \to 0$, we have
\begin{align*}
\pa &\in \Theta(1-p^{\taualpha+1}) &&\pb \in \Theta(p^{\taualpha+1}-p^{\tauzero+1}) \\
\pc &\in \Theta(p^{\tauzero+1}-p^{\tauzeroone+1}) &&\pd \in \Theta(p^{\tauzeroone+1})
\end{align*}
\end{lemma}

In the following subsection, we show that UM and PUM codes outperform decoding in the code $\C$ for small values of $p$.
The result is illustrated in Figure~\ref{fig:PUM_vs_independent}.

\subsubsection{UM Codes}

Let $\Pt$ be the probability that decoding in the UM code is correct in block $t$ and let $\Pindt$ be the probability that decoding in $\C$ is successful.
\begin{theorem}\label{thm:UM_vs_MDS}
There is a $p'>0$ such that $\Pt>\Pindt$ for all $p<p'$.
\end{theorem}

\begin{proof}
Since $\Pt \geq 1-\big( \tfrac{\pc}{1-\pb} \big)^2$ and $\Pindt = \pa+\pc = 1-\pc$, it suffices to show that $\tfrac{\pc}{(1-\pb)^2}<1$ for small $p$.
This directly follows from
\begin{align*}
\tfrac{\pc}{(1-\pb)^2} \in \Theta\left( \tfrac{p^{\tauzero+1}-p^{\tauzeroone+1}}{1-p^{\taualpha+1}}\right) = \Theta\Big( p^{\tauzero+1} \cdot \underset{\in \, \Theta(1)}{\underbrace{\tfrac{1-p^{(\tauzeroone-\tauzero)}}{1-p^{\taualpha}}}}\Big) \to 0 \quad (p \to 0).
\end{align*}
\end{proof}

\subsubsection{PUM Codes}

Now, let $\Pt$ be the probability that decoding in the PUM code is correct in block $t$. We obtain as similar result as for UM codes.
\begin{theorem}\label{thm:PUM_vs_MDS}
There is a $p'>0$ such that $\Pt>\Pindt$ for all $p<p'$.
\end{theorem}

\begin{proof}
By Theorem~\ref{thm:PUM_Pt}, we get $\Pt \geq \pa + \tfrac{\pa}{(1-\pb)^2} \big[\pb(2-\pa-2\pb) + \pa\pc \big]$, so it suffices to show that $\pa + \tfrac{\pa}{(1-\pb)^2} \big[\pb(2-\pa-2\pb) + \pa\pc \big] > \pa+\pb = \Pindt$ for small $p$. This condition can be reformulated into
\begin{align*}
\tfrac{\pb (\pc+\pd)^2}{\pc \pa^2} < 1.
\end{align*}
The left-hand side is
\begin{align*}
\tfrac{\pb \cdot (\pc+\pd)^2}{\pc \cdot \pa^2} \in O\Big( \tfrac{p^{\taualpha+1} \big(p^{\tauzero+1}\big)^2}{\big(p^{\tauzero+1}-p^{\tauzeroone+1}\big)\cdot \big(1-p^{\taualpha+1} \big)} \Big)
\subseteq O\Big( p^{\taualpha+2 \tauzero+1-\tauzero} \Big) \to 0 \quad (p \to 0),
\end{align*}
which proves the claim.
\end{proof}

\begin{figure}[h!]
\input{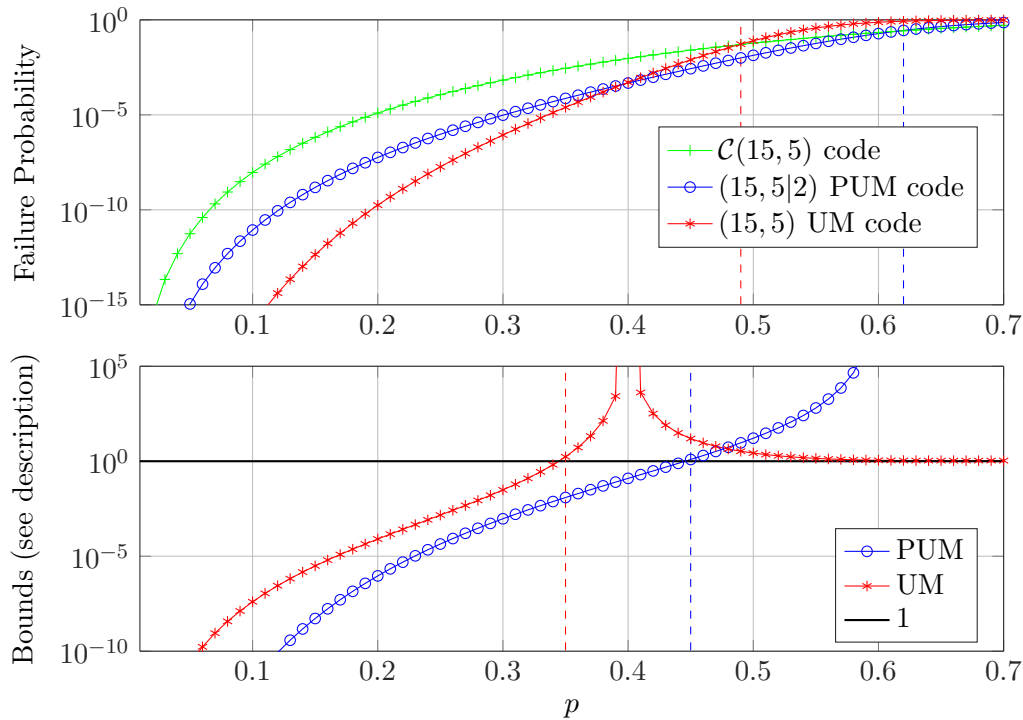}
\caption{Comparison of failure probabilities of a $(15,5 | 2)$ PUM code, a $(15,5)$ UM code, and independent encoding in $\C(15,5)$, in an erasure channel with erasure probability $p$.
The codes are based on MDS codes (i.e., $\tauzero = n-k$), similar to the results in \cite{KB2016}. {\bf Upper half:} Exact failure probability of the PUM, UM, and $\C$ code. The dashed lines indicate the value $p'$, for which $\Pt>\Pindt$ for all $p<p'$ for UM and PUM, respectively. {\bf Lower half:} Upper bounds on $\tfrac{\pc}{(1-\pb)^2}$ (cf.~proof of Theorem~\ref{thm:UM_vs_MDS}) and $\tfrac{\pb \cdot (\pc+\pd)^2}{\pc \cdot \pa^2}$ (cf.~proof of Theorem~\ref{thm:PUM_vs_MDS}) using Theorem~\ref{lem:tail_bounds}. The dashed lines indicate for which $p$ these values drop below $1$, providing lower bounds on $p'$.}
\label{fig:PUM_vs_independent}
\end{figure}

\subsubsection{Adaptions in the Streaming Scenario}

In \cite{KB2016}, PUM codes were used for streaming. In this scenario, in contrast to the assumption above, we do not know $\i_{L}$ and therefore need to decode $\c_{L}$ in $\Calpha$ instead of $\Czero$ in the first decoding step.
In addition, one is interested in reconstructing $\ct$ for $t$ close to $L$ (i.e., the so-called coding delay $L-t$ should be small).
In this case, $\Qt$ remains unchanged.
However, $\Yat$ is always smaller than $\tfrac{\pa}{1-\pb}$, but approaches $\tfrac{\pa}{1-\pb}$ exponentially in $L-t+1$, i.e.,
\begin{align*}
\Yat = \tfrac{\pa}{1-\pb} - \pb^{L-t+1} \tfrac{\pa}{1-\pb} \to \tfrac{\pa}{1-\pb} \quad (L-t+1 \to \infty).
\end{align*}
This also means that for a large coding delay $L-t$, we approach the failure probability values in Theorem~\ref{thm:PUM_Pt} and \ref{thm:UM_Pt} from below, and the results in this section (Theorem~\ref{thm:UM_vs_MDS} and \ref{thm:PUM_vs_MDS}) also hold.

\subsection{Rank-Metric (P)UM Codes in Network Coding}

In \cite{WSS2014} and \cite{PCFBH2015}, (P)UM codes in the rank metric were used for error correction in variants of random linear network coding.
It was observed numerically that (P)UM codes result in lower failure probabilities compared to independent rank-metric codes in this scenario.
The results in this paper might provide a basis for an analytical explanation of this observation.
Although the channel model in \cite{PCFBH2015} is quite complex, it is reasonable that bounds on the tail probabilities of the error weight can be derived, resulting in similar results as in Section~\ref{subsec:PUM_vs_MDS}.

\section{Conclusion}

In this paper, we have derived analytic expressions for the success probability of (P)UM codes in memoryless channels and have shown applications for them.
Besides the already mentioned future work, the results should be generalized---if possible---to certain channels with memory (e.g.\ burst channels).\\

\noindent
{\bf Acknowledgement:} This work was supported by the German Research Foundation (DFG), grant BO~867/29-3.

\end{document}